\documentclass[aps,twocolumn,prl,superscriptaddress,nofootinbib]{revtex4-2}

\usepackage{amsmath}

\usepackage{amssymb}
\usepackage{mathrsfs}
\usepackage{bm, dsfont}
\usepackage[usenames,dvipsnames]{color}
\usepackage{colortbl}
\usepackage[table]{xcolor}
\usepackage{enumitem}

\usepackage{graphicx} 

\usepackage{natbib}
\usepackage[colorlinks=true,linkcolor=blue,citecolor=blue,urlcolor=blue]{hyperref}
\usepackage[capitalise]{cleveref}
\usepackage{hypcap}
\usepackage{verbatim, float}
\usepackage{psfrag}
\usepackage[normalem]{ulem}
\usepackage{physics}		
\usepackage{amsthm}

\newtheorem{theorem}{Theorem}
\newtheorem{definition}{Definition}

\newtheorem{lemma}[theorem]{Lemma}

\definecolor{orangef}{RGB}{210,100,20}

\newcommand{\id}{\mathds{1}}
\newcommand{\ii}{\mathrm{i}}

\begin{document}

\title{Accessible Quantum Gates on Classical Stabilizer Codes}

\author{Victor Barizien}
\affiliation{Institut de Physique Théorique, Université Paris-Saclay, CEA, CNRS, 91191 Gif-sur-Yvette, France}
\author{Hugo Jacinto}
\affiliation{Institut de Physique Théorique, Université Paris-Saclay, CEA, CNRS, 91191 Gif-sur-Yvette, France}
\affiliation{Alice \& Bob, 53 boulevard du Général Martial Valin, 75015 Paris, France}
\author{Nicolas Sangouard}
\affiliation{Institut de Physique Théorique, Université Paris-Saclay, CEA, CNRS, 91191 Gif-sur-Yvette, France}

\date{\today}

\begin{abstract}
  With the advent of physical qubits exhibiting strong noise bias, it becomes increasingly relevant to identify which quantum gates can be efficiently implemented on error-correcting codes designed to address a single dominant error type. Here, we consider $[n,k,d]$-classical stabilizer codes addressing bit-flip errors where $n$, $k$ and $d$ are the numbers of physical and logical qubits, and the code distance respectively.  We prove that operations essential for achieving a universal logical gate set necessarily require complex unitary circuits to be implemented. Specifically, these implementation circuits either consists of $h$ layers of $r$-transversal operations on $c$ codeblocks such that $c^{h-1}r^h \geq d$ or of $h$ gates, each operating on at most $r$ physical qubits on the same codeblock, such that $hr\geq d$.  Similar constraints apply not only to classical codes designed to correct phase-flip errors, but also to quantum stabilizer codes tailored to biased noise. This motivates a closer examination of alternative logical gate constructions using eg.~magic state distillation and cultivation within the framework of biased-noise stabilizer codes. 
\end{abstract}
\maketitle

\paragraph{Introduction—} Fault-tolerant quantum computation refers to the ability to execute large-scale quantum algorithms reliably, despite the presence of errors~\cite{Preskill98}. It relies on quantum error correction codes that encode quantum information across multiple physical qubits, with repeated parity measurements used to identify errors. To process the encoded information, computationally universal gate sets are employed in a way that limits the propagation of errors. The algorithm execution proceeds through a sequence of these gates interleaved with parity checks, continuing until the full computation is completed.

\medskip
\noindent
Among possible error correction codes, stabilizer codes distinguish themselves by parity check operators that are tensor products of single Pauli operators~\cite{Terhal15,gottesman97}. Although they are arguably among the most convenient to implement in practice~\cite{Ryan-Anderson2021,Krinner2022,Acharya2023,Bluvstein2023,Paetznick2024,Google25}, they come with inherent restrictions on the set of gates that can be implemented fault-tolerantly. The Bravyi-Koenig restriction~\cite{Bravyi13} for example, applies to quantum stabilizer codes on a D-dimensional lattice with $D \geq 2$, having a maximum range of the parity checks smaller than the code distance at the power $1/D$.
It states that the restriction to the codespace of any circuit for which the product of the circuit depth and the gate range is smaller than the code distance at the power $1/D$ can only implement unitaries belonging at most to the D$^{\text{th}}$ level of the Clifford hierarchy~\cite{Gottesman1999}.
For the surface code family~\cite{Bravyi98, Kitaev03, Dennis02, Bombin07, Fowler09, Fowler12} ($D=2$ geometry with the support of check operators of diameter $O(1)$), this means that constant-depth circuits made out of local gates can only implement encoded gates within the Clifford group.
Since at least one non-Clifford gate is needed to get a universal gate set, 
fault-tolerance in surface codes cannot be realized using gates admitting constant-depth implementation alone. 
Consequently, alternative strategies, such as magic state distillation~\cite{Bravyi05, Litinsky19} or cultivation~\cite{Knill96, Jones16, Chamberland20, Itogawa24, Gidney24}, are necessary for implementing non-Clifford gates.

\medskip
\noindent
While several refinements and improvements of the Bravyi-Koenig bound have been proposed in the framework of quantum error correction codes~\cite{Pastawski15, Barkeshi24}, analogous impossibility bounds on classical codes have not yet been explored. The value of such bounds is underscored by recent efforts to reduce the overhead of error correction by proposing physical qubit implementations exhibiting biased noise~\cite{Mirrahimi14, Puri17, Guillaud19, Puri20, Guillaud21, Chamberland22, Xu22}. Specifically, dissipatively stabilized cat qubits exhibit a bit-flip error rate that is exponentially suppressed with the average number of photons, at the cost of a linearly increasing phase-flip error rate. This stark asymmetry leads to bit-flip errors that are many orders of magnitude less frequent  than phase-flip errors~\cite{Lescanne20, Berdou23, Marquet24}. As a result, long gate sequences can be executed reliably by focusing error correction solely on the dominant type of error. This reduces significantly the physical qubit number with respect to standard approaches driven by quantum codes to execute quantum algorithms~\cite{Gouzien23}. Recent experimental results~\cite{Putterman2025} demonstrating a logical qubit memory based on cat qubits, protected by an outer repetition code of distance $d = 5$  underscore the practical viability of concatenated bosonic codes. These results also highlight the pressing need to identify the set of quantum gates that can be implemented efficiently within classical stabilizer codes, in order to guide upcoming experimental efforts and solidify concatenated bosonic codes as a promising architecture for fault-tolerant quantum computation.

\medskip 
\noindent 
Here, we focus on stabilizer codes that address one type of error, say bit-flip (X-type) error for concreteness (the same results apply to phase-flip errors). We prove that the only logical operations that can be done efficiently on these classical codes are the ones that evolve every $Z$ operator into a sum of products of $Z$ operators. Other operations — essential for achieving a universal logical gate set — can still be performed at the logical level, 
but their implementation necessarily requires complex unitary circuits made either with $h$ layers of $r$-transversal operations operating on $c$ codeblocks such that $c^{h-1} r^h \geq d$ where $d$ is the code distance, or a sequence of $h$ gates, each acting on at most $r$ qubits on the same codeblock, with the constraint  $hr \geq d$. Extensions of these bounds to quantum stabilizer codes designed to correct biased-noise are discussed. These practical constraints motivate a closer examination of constructions of logical gates not relying on unitary circuits, such as magic state distillation and cultivation within the framework of biased-noise stabilizer codes. 

\bigskip 
\noindent
\paragraph{Classical stabilizer codes—} Classical stabilizer codes can be characterized by a set of parameters $[n,k,d]$, where $n$ is the number of physical qubits, $k$ the number of logical qubits and $d$ the distance of the code. Each code is generated by a set of $\varsigma$ commuting stabilizers $\{\sigma_t\}_{1\leq t \leq \varsigma}$ that are products of Pauli $Z$ operators only, such that the codespace $\mathcal{C}$ is given by the intersection of the $+1$ eigenspace of all stabilizers. In general, the stabilizers are taken to be independent -- none of them is a product of any others -- implying that $k = \log_2 (\dim \mathcal{C}) = n-\varsigma$, and the codespace is isomorphic to a logical space $\mathcal{E}=(\mathbb{C}^2)^{\otimes k}$.

\medskip
\noindent
We say that a unitary operation $U$ on the physical space gives a logical operation if it preserves the codespace, i.e.~that $U(\mathcal{C}) \subset \mathcal{C}$. In this case, we denote the corresponding logical operation with respect to a given isometry $S:\mathcal{E}\to\mathcal{C}$ to be 
$U_L = S^\dagger U S$. 
Note that a given logical unitary operation $U_L$ can admit several implementations $U$ at the physical level.

\bigskip 
\noindent
\paragraph{Logical Pauli implementations—}
A class of logical operations of particular interest are those admitting an implementation as a product of physical Pauli operators. 
Let us first consider implementations involving products of Pauli $Z$ operators. For a given $\alpha \in \{0,1\}^n$, we denote by $Z_\alpha = \bigotimes_{i=1}^n (Z^{(i)})^{\alpha_i}$ the physical operator acting with Pauli $Z$ on each qubit $i$ for which $\alpha_i=1$, and the identity elsewhere. Since the stabilizers $\sigma_t$ involve only products of Pauli $Z$, $Z_\alpha$ commutes with every stabilizer and as such preserves the codespace $\mathcal{C}$. Thus, $Z_{L,\alpha}:=SZ_\alpha S^\dagger$ is a valid logical operator. 

\medskip 
\noindent
Note that the isomorphism $S$ is not unique but a canonical expression can be chosen from the parity-check matrix of the associated classical code, see End~Matter~A. In particular, this choice, up to a relabeling of the physical qubits, ensures that the individual physical $Z$ operators acting on the first $k$ qubits correspond directly to the logical $Z$ operators on the $k$ qubits of the logical space $\mathcal{E}$.
In other words, $Z_{L}^{(j)} := Z_{L,(\delta_{ij})_i} = SZ^{(j)}S^\dagger$, for $1\leq j\leq k$, corresponds to a Pauli $Z$ on the $j$-th logical qubit for this canonical isometry, see End~Matter~B. 

\medskip 
\noindent
Let us now consider the logical operators obtained from implementation involving products of Pauli $X$ operators.
From the classical definition of the code distance $d$ as the minimal non-zero Hamming weight among the codewords, $d$ corresponds to the minimal number of bit-flips, that is physical Pauli $X$ operations, that must be applied to transform one logical codeword into another without being detectable by the stabilizers. 
This implies that any products of strictly less than $d$ physical $X$ operations would not preserve the codespace. Note that this definition is a coherent adaptation of the distance of quantum codes (the minimum number of physical qubits that need to be acted upon to produce a non-trivial logical operation) for cases where only $X$-errors are detected. This definition of the code distance implies that
\begin{equation}
    \forall I \subset [n] \text{ s.t } |I|<d, \  S^\dagger \left( \bigotimes_{i\in I} X^{(i)}\right) S = 0_L \label{eq:Xrule}   
\end{equation}
for any choice of isometry $S$.

\bigskip 
\noindent
\paragraph{Efficient logical operations—} In the following, we consider a classical stabilizer code defined on a lattice of physical qubits. We consider $c$ codeblocks of such a code, which are characterized by the parameters $[cn,ck,d]$ as each codeblock is associated to the parameters $[n,k,d]$.
The physical Hilbert space $\mathcal{H}$ is thus naturally partitioned into the tensor product of $c$ Hilbert spaces $\mathcal{H} =\bigotimes_{m=1}^{c} \mathcal{H}_m$ each containing $n$ physical qubits. 

\medskip
\noindent
Let's consider an arbitrary logical operation $U_L$ and suppose that it admits an implementation $U$. Our aim is to compute the commutation relation of the logical operation $U_L$ and an arbitrary logical operator $Z_{L}^{(j)}$. From its definition, $Z_L^{(j)}$ can be implemented from a single physical $Z$ operators on the $j$-th physical qubits, that is $Z_L^{(j)} = S^\dagger Z^{(j)} S$. We consider the operator $K=UZ^{(j)} U^\dagger$. The evolution $Z^{(j)} \to U Z^{(j)} U^\dagger$ enlarges the support of $Z^{(j)}$ from $1$ to a given support which depends on the way $U$ is implemented. We will consider different implementations for $U$ below but for now, we simply denote by $\delta_m$
the support of $K$ on each codeblock, i.e.~$\delta_m = \operatorname{supp}(U Z^{(j)} U^\dagger)\cap \mathcal{H}_m$. We prove now that if $|\delta_m|<d$ for all $m$, then $U_L Z_L^{(j)}U_L^\dagger$ is a linear combination of product of logical Pauli $Z$ operators and the identity.

\medskip
\noindent
Indeed, for all $l$ the set of tensor products of all possible~$l$ Pauli operators (including identity) forms a real orthogonal basis for the vector space of $2l \times 2l$ hermitian matrices, and thus the operator $K$ can be decomposed as $K = \sum_\alpha b_\alpha K_\alpha$, where $b_\alpha$ are real numbers and $K_\alpha$ are products of Pauli operators and the identity with support $\delta_m$ on each codeblock. Using the fact that $U(\mathcal{C}) \subset \mathcal{C}$, we have $S^\dagger U = U_L S^\dagger$, and thus $S^\dagger K S = U_L Z_L^{(j)} U_L^\dagger$. However $S^\dagger K S = \sum_\alpha b_\alpha S^\dagger K_\alpha S$. If we assume that, for all $m$, $|\delta_m| < d$, no term $K_\alpha$ has support on at least $d$ physical qubits in the same codeblock and thus \cref{eq:Xrule} implies that no logical $X$ nor logical $Y$ appears in the sum. As such, one has $U_L Z_L^{(j)} U_L^\dagger = \sum_\alpha b_\alpha Z_{L,\alpha}$.

\medskip
\noindent
If this property holds true for all possible $j$, then for any product $Z_{L,\alpha}$, one has $U_L Z_{L,\alpha} U_L^\dagger = \sum_{\alpha'} b'_{\alpha'} Z_{L,\alpha'}$. In other words, if one denotes by $V_Z$ the algebra spanned by the $Z_{L}^{(j)}$, $U_L$ must satisfy $U_L V_Z U^{\dag}_L \subset V_Z$. We will use the contrapositive of this statement to derive lower bounds on the complexity of circuits implementing logical gate operations that do not preserve the structure of products of Z and identity operators under conjugation.

\medskip
\noindent
Note that equivalent derivations can be made in more general situations. First, each codeblock can be built out of different codes and have different parameters $[n_m,k_m,d_m]$ for each $m\in[c]$, as long as $n=\sum_m n_m$, $k=\sum_m k_m$ and $d=\text{min}_m~d_m$. Second, the physical circuit $U$ can map the codespace $\mathcal{C}$ to the codespace of another code $\mathcal{C}_2$. In this case, one can choose isometries $S_1$ and $S_2$ for each code.
An implementation $U$ of $U_L$ must satisfy $U(\mathcal{C}_1) \subset \mathcal{C}_2$ and $U_L = S_2^\dagger U S_1$. Since any logical $Z_L^{(j)}$ operation can be implemented with a physical Pauli $Z^{(j)}$ of support $1$ on the first code, a similar derivation shows that $U_L V_Z U^{\dag}_L \subset V_Z$ still holds as long as the support of $UPU^\dagger$ on each codeblock of the second code $\delta_m$ satisfies $|\delta_m|<d$ for all $m$.

\medskip
\noindent
In the following paragraphs, we consider different choices for possible physical implementations $U$ of logical gates $U_L$. We show how the property derived above translates into fine-grained bounds on the complexity of these physical implementations.

\begin{figure}
    \centering
    \includegraphics[width=0.8\linewidth]{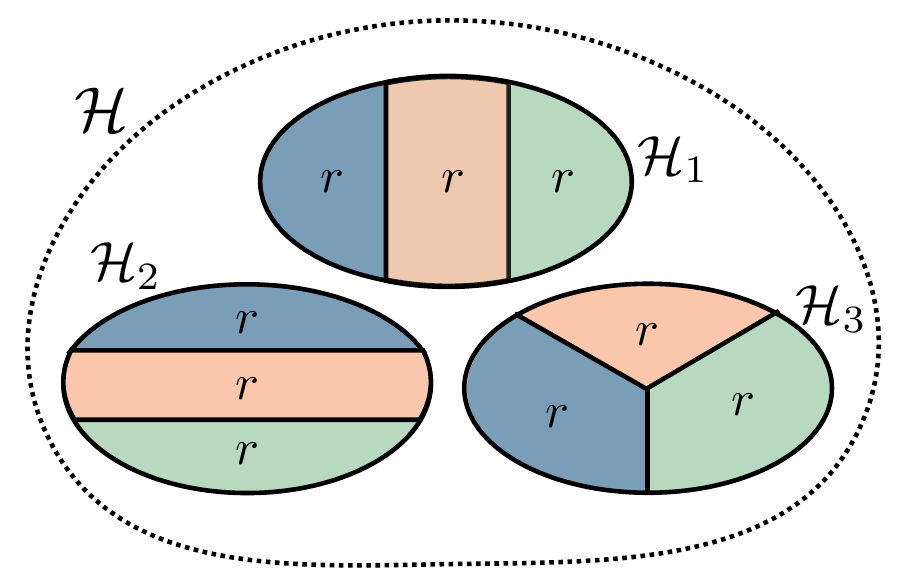}
    \caption{Illustration of a $r$-transversal partition of the physical space $\mathcal{H}$. The partition is given by three disjoint sets of physical qubits as $\mathcal{H}=\mathcal{G}_{\text{blue}}\sqcup \mathcal{G}_{\text{orange}} \sqcup \mathcal{G}_{\text{green}}$. The $c$ different codeblocks correspond to the circle area labelled by $\mathcal{H}_m$, and it holds that $|\mathcal{G}_{q} \cap \mathcal{H}_m|\leq r$, for all $m\in [c]$ and $q\in\{blue,green,orange\}$. An implementation $U$ is called $r$-transversal is there exists a $r$-transversal partition for which one can write $U=\bigotimes_q U_q$, where $\operatorname{supp}(U_q) \subset \mathcal{G}_q$. In this example, it would mean that $U=U_{\text{blue}} \otimes U_{\text{orange}} \otimes U_{\text{green}}$.}
    \label{fig:enter-label}
\end{figure}

\bigskip
\noindent
\paragraph{r-transversal gates—}
We here consider implementations $U$ made out of $r$-transversal operations. We start by reviewing the definitions of $r$-transversality, which are illustrated in~\cref{fig:enter-label}.

\begin{definition}[$r$-transversal partition]
    A $r$-transversal partition of the physical space is a partition of $\mathcal{H}$ into sets each containing at most $r$ qubits of each codeblock.
\end{definition}
\begin{definition}[$r$-transversal operation]
    A physical implementation $U$ is $r$-transversal if there exists a $r$-transversal partition of $\mathcal{H}$ such that $U$ only involves gates with support within the same set of the partition.
\end{definition}
\noindent
The case $r=1$ corresponds to transversal operations, that is, gates operating on a disjoint ensemble of qubits, with each gate acting on at most one qubit on each codeblock.  For example, a transversal CNOT gate between pairs of qubits belonging to two codeblocks of a repetition code is an example a 1-transversal operation. Transversal operations are of particular interest in the context of fault-tolerant quantum computation because gates having no share support do not spread Pauli errors within a codeblock. They can also be executed in parallel, hence avoiding noise on idle qubits.

\medskip
\noindent
If $U$ has a $r$-transversal implementation, it can be decomposed as a tensor product $U=\bigotimes_q U_q$ where each $U_q$ only acts within the same set of a $r$-transversal partition. For the rest of the paper, we label $P$ any physical $Z^{(j)}$ operator acting on a single qubit $j$. Since the support of $P$ is 1, there exists only one index $q$, labelled $q^\star$, such that $\operatorname{supp}(U_q)\cap\operatorname{supp}(P) \neq \varnothing$. As such, one has $UPU^\dagger = U_{q^\star} P  U_{q^\star}^\dagger$. Since the support of $U_{q^\star}$ in each codeblock is at most $r$, the support of the evolution $P\to UPU^\dagger$ can be upper bounded as
\begin{equation}
\label{bound_rtransverse}
    \forall m \in [c], ~|\delta_m|\leq  r .
\end{equation}

\bigskip
\noindent
\paragraph{$h$ layers of r-transversal gates—}
We now consider a physical implementation of the form $U = \prod_{\ell=1}^h U_\ell$ where $U_\ell$ are $r$-transversal operations, that is $h$ layers of $r$-transversal operations. Note that each layer may involve a different $r$-transversal partition of the physical space. In any case, the support of $P \rightarrow U P U^{\dagger}$ within each codeblock is upper bounded as 
\begin{equation}
        \forall m \in [c], ~|\delta_m|\leq  c^{h-1}r^h 
\end{equation} 
where $c$ refers to the number of codeblocks. This can be shown by induction on $h$. The case $h=1$ holds according to Eq.~\eqref{bound_rtransverse}. We assume that the result holds for implementations involving $h$ layers of $r$-tranversal gates. Now, consider an implementation with $h+1$ layers $U=\prod_{\ell=1}^{h+1} U_\ell.$ Since one has $UPU^\dagger = U_{h+1} \tilde{U}P\tilde{U}^\dagger U_{h+1}^\dagger$, where $\tilde{U}$ has $h$ layers of $r$-transversal gates, it suffices to look at how $U_{h+1}$ changes the support of $\tilde{U}P\tilde{U}^\dagger$ within each codeblock. Let's denote $\delta_{h}$ the support of $\tilde{U}P\tilde{U}^\dagger$, and $\delta_{m,h}$ its intersection with the codeblock $\mathcal{H}_m$. The worst-case scenario for support spreading is if every qubits in $\delta_h$ belongs to a different set of the partition associated to $U_{h+1}$. In this case, every qubits will be spread by the $U_{h+1}$ evolution to at most $r-1$ qubits in their own codeblock $m$.
Any other codeblock $m'\neq m$ will bring at most $r$ new qubits for each of the $\delta_{m',h}$ qubits in codeblock $m$.
In the end, we have
\[ |\delta_{m,h+1}|\leq  |\delta_{m,h}| + |\delta_{m,h}| (r-1)+\sum_{m'\neq m}|\delta_{m',h}| \cdot r. \]
Since the result holds for $h$ layers, we have for all $m$, $|\delta_{m,h}| \leq c^{h-1}r^h$, and as such
\[|\delta_{m,h+1}| \leq c\cdot r \cdot c^{h-1}r^h = c^{h}r^{h+1}.\]

\medskip
\noindent
\paragraph{Consecutive gates—}
We finally consider the case where $U$ is implemented with a sequence of $h$ successive gates $U = \prod_{\ell=1}^h U_\ell$, each gate $U_\ell$ having a support on at most $r$ qubits of each codeblock. Since every gate is applied successively, they can share support on common qubits. In this scenario, the support of $UPU^\dagger$ in each codeblock is enlarged by at most $r$ qubits at each layer,  so that 
\begin{equation}\label{conseutive gate constraint}
    \forall m \in [c], ~|\delta_m|\leq hr.
\end{equation}
Note that the first gate $U_1$ can actually be a $r$-transversal operation without modifying the bound as in the first layer, the support of $P$ is 1 in only one codeblock.

\bigskip 
\noindent
\paragraph{No-go theorem for classical stabilizer codes—}
The converse of the property described above together with the bound on the support of $UPU^\dagger$ leads to the following theorem: 
\begin{theorem} \label{lem:gateproperties}
    Consider $c$ codeblocks of a classical error-correcting stabilizer code that accounts for $X$ type errors with parameters $[n,k,d]$, together with an isometry $S:\mathcal{E}\to \mathcal{C}$. 
    Then, for any logical operator $U_L$, if there exists $j\in[n]$ such that 
    \begin{equation}\label{eq: not preserving Vz}
    U_L Z_L^{(j)} U_L^\dag \neq \sum_\alpha b_\alpha Z_{L,\alpha}
    \end{equation}
    where $Z_L^{(j)}:= S Z^{(j)} S^\dagger$ is the logical operator obtained from a Pauli Z on physical qubit $j$, $b_\alpha$ are real numbers and $Z_{L,\alpha}$ are products of logical Pauli $Z$ operators, then $U_L$ cannot be implemented at the physical level:
    \begin{enumerate}
        \item with a circuit of $h$ layers of $r$-transversal operations such that $c^{h-1}r^h<d$.
        \item nor with a circuit of $h$ gates, each operating on at most $r$ physical qubits in each code block, such that $hr<d$.
    \end{enumerate}
\end{theorem}

\noindent
One can also directly study the size of the support of the implementation $U$ on each codeblocks. When using the decomposition of $U$ into Pauli strings and looking at the corresponding logical operations, it comes that if $U_L$ itself cannot be decomposed as a linear combination of logical $Z_{L,\alpha}$ operators, it cannot be implemented with a sequence of consecutive gates such that $hr<d$. This is more restrictive than \cref{lem:gateproperties}, point 2. 

\medskip 
\noindent
Naturally, a similar theorem can be formulated for classical stabilizer codes accounting for $Z$ errors, by simply swapping $X$ and $Z$ operators in the previous derivations. 

\medskip
\noindent
Interestingly, \cref{lem:gateproperties} can be extended to quantum stabilizer codes accounting for both bit-flip and phase-flip errors.
In this case, the two inequalities of 1.\@ and 2.\@ respectively become $d_Z c^{h-1}r^h<d_X$  and $(d_Z-1) + hr<d_X$ where $d_Z$ and $d_X$ are the code distances for $Z$ and $X$ errors. Those bounds are trivial for balanced error corrections ($d_Z=d_X$) but they remain very constraining for qubits subject to biased noises for which relevant codes satisfy $d_Z\ll d_X$.

\medskip 
\noindent 
When choosing the canonical isometry for the encoding, operators $Z_L^{(j)}$ correspond to Pauli-$Z$ on the $j$-th logical qubit. In this case, examples of gates that fall under the no-go of~Theorem~\ref{lem:gateproperties} include $H$, $\sqrt{X}$ and $\sqrt{\sqrt{X}}$ since they verify $H Z H^\dagger = X$, $\sqrt{X} Z \sqrt{X}^\dagger = -Y$ and $\sqrt{\sqrt{X}} Z \sqrt{\sqrt{X}}^\dagger = \cos{\pi/4} \id + \ii \sin{\pi/4} Y$. More generally, this holds for every operator of the form $\sqrt[q]{X}$ for $q\geq 1$.

\medskip 
\noindent
Note that the property to evolve any $Z$ operator to a sum of products of $Z$ operators is preserved under composition. Thus, the operators that can be implemented with either a circuit of $h$ gates, each operating on at most $r$ physical qubits, such that $hr<d$ or with a circuit of $h$ layers of $r$-transversal operations such that $c^{h-1}r^h<d$ form a strict subgroup of all possible operators, and thus cannot lead to a universal gate set, see End Matter~C.

\bigskip
\paragraph{Conclusion—} 
We have established strong constraints on the implementation of universal logical operations on classical stabilizer codes  and biased-noise quantum stabilizer codes. These results offer valuable guidance for the design of fault-tolerant gates on stabilized cat qubits protected by classical codes. When this approach is viewed through the lens of concatenated bosonic codes, the Eastin-Knill theorem~\cite{Eastin09} applies, ruling out the possibility of realizing a universal set of transversal gates. However, our results go further by providing a much more fine-grained characterization. They identify precisely which gates are hard to implement with unitary circuits and quantify, exactly rather than asymptotically, the circuit size required to realize them in various scenarios, including r-transversal layers and multi-qubit gate sequences. In particular, we show that relaxing strict transversality to r-transversality offers little advantage in terms of fault tolerance: achieving a universal gate set still requires exceeding the code's capacity. This highlights a clear need to explore alternative approaches to logical gate implementation, such as magic state distillation, within the framework of classical stabilizer codes and biased-noise quantum stabilizer codes.

\bigskip
\paragraph{Acknowledgments—} 
We thank E. Gouzien, J. Guillaud, and D. Ruiz for insightful discussions at an early stage of the project, and V. Savin for suggesting extensions to biased-noise codes and, more broadly, for providing valuable feedback on an earlier version of the manuscript. We acknowledge funding by Agence Nationale de la Recherche in the framework of France 2030 with the reference ANR-22-PETQ-0007 and project name EPiQ. This work was also partially supported by the French
National program Programme d’investissement d’avenir, IRT Nanoelec, with the reference ANR-10-AIRT-05.

\bibliographystyle{apsrev4-2}
\bibliography{sample}

\onecolumngrid
\ \\ \newpage
\section*{End Matter}
\twocolumngrid

\section*{A: Isometry choice for classical stabilizer codes}
\label{app:isometry_choice}

The isomorphism between the logical space $\mathcal{E}$ and the codespace $\mathcal{C}$ is not unique but a convenient expression can be chosen from the parity-check matrix of the associated classical code. Indeed, the stabilizer codes of interest account for bit-flip errors and thus involve only $Z$ type stabilizers that can be written as 
\begin{equation}
    \sigma_t = \bigotimes_{i = 1}^n (Z^{(i)})^{(h_t)_i}
\end{equation}
where $h_t \in \mathbb{F}_2^n$ is a binary vector that encodes whether the operator $Z^{(i)}$ appears in the product or not. As such, there is a canonical isomorphism between the stabilizers group $\langle\sigma_t,\cdot\rangle$ and the binary vectors group $\langle h_t, \oplus \rangle$. Since the stabilizers are independent, the set $h_t$ itself forms a free family of $\mathbb{F}_2^n$. They can thus be used to build the parity-check matrix $H$ of a classical linear code as the $(n-k)\times n$ matrix whose lines are made of the vectors $\{h_t\}_{1\leq r \leq n-k}$. This matrix is full-ranked and can therefore be written in the standard form $H  = [-P^T | I_{n-k}] \in \mathbb{F}_2^{(n-k)\times n}$ where $P^T$ is a $(n-k)\times k$ matrix and $I_{n-k}$ the $(n-k)$ identity matrix~\cite{Ling04}. 
The generating matrix of this code is given by $G=[I_k|P]\in \mathbb{F}_2^{k\times n}$ such that the following relation holds
\begin{equation} \label{eq:orthogonality}
    \bigoplus_{i=1}^n (g_s)_i (h_t)_i = 0,
\end{equation} 
for all $s\in[k], t\in[n-k]$ where $\{g_s\}_{1\leq s \leq k}$ denote the lines of $G$. Note that in the theory of classical linear codes, the distance $d$ of a code is given either by 
the minimal non-zero Hamming weight of the codewords.

\medskip 
\noindent 
The physical encoding of the logical words
$|\epsilon_1\rangle \otimes \hdots \otimes |\epsilon_k\rangle$ with $\epsilon_s\in \{0,1\}$ of the stabilizer code is specified by applying $G$ on $|\epsilon_1\rangle \otimes \hdots \otimes |\epsilon_k\rangle$, specifically
\begin{equation}
\label{Gepsilon}
    \ket{G_\epsilon} = \bigotimes_{i=1}^n \ket{\bigoplus_{s=1}^k \epsilon_s (g_s)_i}
\end{equation}
where $\epsilon$ denotes the vector $(\epsilon_s)\in \{0,1\}^k$. From the choice of the parity-check matrix, we have $(g_s)_i = \delta_{si}$ for any $1 \leq i\leq k$. Therefore, the vectors $\ket{G_\epsilon}$ are of the form
\begin{equation}
\label{Gepsilonbis}
    \ket{G_\epsilon} = \ket{\epsilon_1}\otimes...\otimes\ket{\epsilon_k} \bigotimes_{i=k+1}^n \ket{\bigoplus_{s=1}^k \epsilon_s (g_s)_i}.
\end{equation}
The $2^k$ codewords $\ket{G_\epsilon}$ form a basis of the codespace $\mathcal{C}$, that is they are linearly independent and $+1$ eigenstates of the stabilizers. Indeed, one can verify that $\langle G_{\epsilon}|G_{\epsilon'}\rangle = \delta_{\epsilon \epsilon'}$ from Eq.~\eqref{Gepsilonbis} and that the action of a stabilizer $\sigma_t$ is given by 
\begin{equation}
\nonumber
\begin{split}
   \sigma_t \ket{G_\epsilon} & = \left[\bigotimes_{i=1}^n (Z^{(i)})^{(h_t)_i}  \right] \bigotimes_{i=1}^n \ket{\bigoplus_{s=1}^k \epsilon_s (g_s)_i} \\
    & = \bigotimes_{i=1}^n (-1)^{(h_t)_i \left(\bigoplus_{s=1}^k \epsilon_s (g_s)_i\right)}  \ket{\bigoplus_{s=1}^k \epsilon_s (g_s)_i} \\
    & = (-1)^{\bigoplus_{i=1}^n (h_t)_i \left(\bigoplus_{s=1}^k \epsilon_s (g_s)_i\right)} \bigotimes_{i=1}^n \ket{\bigoplus_{s=1}^k \epsilon_s (g_s)_i} \\
    & = (-1)^{\bigoplus_{s=1}^k \epsilon_s \left( \bigoplus_{i=1}^n  (g_s)_i (h_t)_i \right)} \ket{G_\epsilon} \\
    & = \ket{G_\epsilon} 
\end{split}
\end{equation}
for any $t\in[n-k]$, $\epsilon \in \{0,1\}^k$, where we have used \cref{eq:orthogonality}. 

\medskip 
\noindent 
This sets a simple form of the isomorphism from the logical space $\mathcal{E}$ to the code space $\mathcal{C}$ as 
\begin{equation} \label{eq:isometry}
    S = \sum_{\epsilon_1,\hdots,\epsilon_k} \ket{G_{\epsilon}}\bra{\epsilon_1}\otimes...\otimes\bra{\epsilon_k}.
\end{equation}
One can verify that $S$ is an isometry -- $S^\dagger S = \id_\mathcal{E}$ and $SS^\dagger$ is the orthogonal projection onto $\mathcal{C}$. For a classical stabilizer code, we call the isometry $S$ given in \cref{eq:isometry} the \textit{canonical} isometry of the code. 

\medskip
\noindent
In all generality, the isometry between the logical space and the codespace is given by the canonical isometry $S$ up to a change of basis on both physical  and logical spaces. Therefore, one can write it in its most general form to be $S' = VSW^\dagger$ where $V$ and $W$ denote the basis change operations in the physical and logical spaces respectively. In order for $S'$ to be a valid isometry, it must verify that $S'S'^\dagger$ is the projection onto the codespace, and therefore that $V(\mathcal{C}) \subset \mathcal{C}$. This means that $V$ must preserve the codespace, and therefore, if we write $V|_\mathcal{C}=\sum_{\epsilon,\epsilon'} V_{\epsilon \epsilon'} \ket{G_\epsilon}\bra{G_{\epsilon'}}$ its restriction to the codespace, one has 
\begin{equation}
\label{eq:changeofbasis}
S' = VSW^\dagger = V|_{\mathcal{C}}SW^\dagger = S\Tilde{W}^\dagger
\end{equation}
where we introduced $\tilde{W} := W \cdot \sum_{\epsilon,\epsilon'} V_{\epsilon \epsilon'} \ket{\epsilon} \bra{{\epsilon'}}$. As such, any physical change of basis in the isometry can be properly regarded as a change of logical basis.

\section*{B : Logical Pauli operators for the canonical isometry} \label{app:isometry_choice}

Considering the canonical isometry $S$ and any  $1\leq j \leq k$, the $Z$-Pauli operator on logical qubit $j$ admits an implementation $Z^{(j)}$ at the physical level from a single $Z$-Pauli on qubit $j$, as we have
\begin{equation}
\nonumber
    \begin{split}
            S^\dagger Z^{(j)} S \ket{\epsilon_1} \otimes...& \otimes\ket{\epsilon_k}  = S^\dagger Z^{(j)} \ket{G_\epsilon} \\ 
            & = S^\dagger Z^{(j)} \bigotimes_{i=1}^n \ket{\bigoplus_{s=1}^k \epsilon_s (g_s)_i} \\
            & = S^\dagger (-1)^{\bigoplus_{s=1}^k \epsilon_s (g_s)_j} \bigotimes_{i=1}^n \ket{\bigoplus_{s=1}^k \epsilon_s (g_s)_i} \\
            & = (-1)^{\bigoplus_{s=1}^k \epsilon_s (g_s)_j} \ket{\epsilon_1}\otimes...\otimes\ket{\epsilon_k}\\
            & = (-1)^{\epsilon_j} \ket{\epsilon_1}\otimes...\otimes\ket{\epsilon_k} 
    \end{split}
\end{equation}
where we used $(g_s)_j = \delta_{sj}$ for $1\leq j \leq k$. 
This implies that any product of logical $Z$ operators can be implemented by a product of physical $Z$ operators with support on at most $k$ qubits. 
Note that, reciprocally, the action of the physical Pauli $Z^{(i)}$ operator on qubit $i$ implements the logical operation acting as a Pauli-$Z$ on each qubits for which $(g_j)_i =1$.

\medskip
\noindent
Finally, one implementation of the logical operator performing a Pauli-$X$ operation on the $j$-th logical qubit is given by the $j$-th line of the generating matrix $G$, as the operation applying a physical $X$ on each qubits $i$ for which $(g_j)_i=1$.

\section*{C: Non-universality of gates falling under Theorem~\ref{lem:gateproperties}} \label{app:not_universal}

For any other isometry choice $S'=SW^\dagger$, the corresponding logical operations $U_L'$ are equivalent to $U_L$ obtained for the canonical isometry up to a basis change, see~\cref{eq:changeofbasis}. As such, a given $U_L'$ with an implementation $U$ at the physical level for which the support of $UPU^\dagger$ is upper bounded by $d$ for any physical single Pauli Z operation $Z^{(j)}$ satisfies $U_L' W Z_L^{(j)} W^\dagger U_L'^\dagger = \sum_\alpha b_\alpha W Z_{L,\alpha} W^\dagger$, that is $U_L'(WV_ZW^\dagger) \subset WV_ZW^\dagger$ for the algebra $V_Z$ spanned by $Z_L^{(j)}$ corresponding to the canonical isometry. As such, to study universality, it suffices to focus on the canonical isometry, for which $V_Z$ is the subspace of all product of Pauli-$Z$ operations.

\begin{lemma}
    For a given operator $U_L$ acting on an Hilbert space of $k$ qubits, one has: 
    \begin{equation}
        \left(\forall j\in[k], \   U_L Z^{(j)}_L U_L^\dagger = \sum_\alpha b_\alpha Z_{L,\alpha} \right) \Longleftrightarrow U_L V_Z U_L^\dagger \subset V_Z
    \end{equation}
    where $V_Z$ denotes the algebra spanned by the $Z^{(j)}_L$ operators. 
\end{lemma}

\begin{lemma}\label{lem:closedsubgroup}
    The subgroup $G$ of unitary operators verifying $U_L V_Z U_L^\dagger \subset V_Z$ 
    is closed, and thus not universal. 
\end{lemma}

\begin{proof}
    We consider the standard operator norm $|||.|||$, which is multiplicative, i.e.~$|||A\cdot B|||\leq|||A|||\times |||B|||$. Note that for any operator $U_L$, we have $|||U_L\dagger|||=|||U_L|||$ and $|||U_L|||=1$ iff $U_L$ is unitary. 
    
    Now we consider an arbitrary unitary operator $U_L$ that belongs to $\bar G$. For every $\varepsilon>0$, there exists an operator $g_\varepsilon \in G$ such that $|||U_L-g_\varepsilon|||\leq \varepsilon$. For every $P_L\in V_Z$ there exists $P_L' \in V_Z$ such that $g_\varepsilon P_L g_\varepsilon^\dagger = P_L'$. Then, it follows
    \begin{widetext}
    \begin{equation}
    \begin{split}
        |||P_L'-U_L P_L U_L^\dagger||| & = |||g_\varepsilon P_L g_\varepsilon^\dagger-U_L P_L^\dagger U_L^\dagger||| \\
        & \leq |||g_\varepsilon P_L g_\varepsilon^\dagger-U_L P_L g_\varepsilon^\dagger||| + |||U_LP_L g_\varepsilon^\dagger-U_L P_L U_L^\dagger||| \\
        & \leq |||P_L||| \times (|||g_\varepsilon^\dagger|||\times |||g_\varepsilon-U_L||| + |||U_L|||\times|||g_\varepsilon^\dagger-U_L^\dagger|||) \\
        & \leq 2\varepsilon
    \end{split}
    \end{equation}
    \end{widetext}
    Taking $\varepsilon \to 0$ in the above shows that $U_L P_L U_L^\dagger \in \Bar{V_Z}$. Yet, $\Bar{V_Z}=V_Z$ and thus $U_L P_L U_L^\dagger \in V_Z$ for all $P_L$, which is exactly $U_L\in G$. In the end, we proved that $G$ is closed, and since it does not contain every operators of $L(\mathcal{E})$ -- such as $H$ for example -- it cannot form a universal set of gates.
\end{proof}

\clearpage
\end{document}